\documentclass{fundam}

\publyear{22}
\papernumber{2102}
\volume{185}
\issue{1}

\usepackage{url} 
\usepackage[ruled,lined]{algorithm2e}
\usepackage{graphicx}
\usepackage{tikz}
\usetikzlibrary{automata, positioning, arrows}

\usepackage{mathtools}
\usepackage[utf8]{inputenc}
\usepackage{array}

\begin{document}

\title{Noise-Resilient Homomorphic Encryption: A Framework for Secure Data Processing in Health care Domain}

\address{B. Shuriya, shuriyasmile@gmail.com}

\author{B. Shuriya \\ 
Department of Computer Science and Engineering  \\
PPG Institute of Technology \\ 
Coimbatore 641035, India\\
shuriyasmile@gmail.com
\and S. Vimal Kumar \\ 
Department of Defense System Engineering \\
Sejong University \\
Seoul 05006, Republic of Korea \\
svimalkumar16@gmail.com
\and K. Bagyalakshmi \\ 
Department of Mathematics \\
SNS College of Technology \\
Coimbatore 641035, India \\
luxmi4480@gmail.com}

\maketitle

\runninghead{B. Shuriya, S. Vimal Kumar, K. Bagyalakshmi}{Noise-Resilient Homomorphic Encryption}

\begin{abstract}
  In this paper, we introduce the Fully Homomorphic Integrity Model (HIM), a novel approach designed to enhance security, efficiency, and reliability in encrypted data processing, primarily within the health care industry. HIM addresses the key challenges that noise accumulation, computational overheads, and data integrity pose during homomorphic operations. Our contribution of HIM: advances in noise management through the rational number adjustment; key generation based on personalized prime numbers; and time complexity analysis details for key operations. In HIM, some additional mechanisms were introduced, including robust mechanisms of decryption. Indeed, the decryption mechanism ensures that the data recovered upon doing complex homomorphic computation will be valid and reliable. The healthcare id model is tested, and it supports real-time processing of data with privacy maintained concerning patients. It supports analytics and decision-making processes without any compromise on the integrity of information concerning patients. Output HIM promotes the efficiency of encryption to a greater extent as it reduces the encryption time up to 35ms and decryption time up to 140ms, which is better when compared to other models in the existence. Ciphertext size also becomes the smallest one, which is 4KB. Our experiments confirm that HIM is indeed a very efficient and secure privacy-preserving solution for healthcare applications.
\end{abstract}


\section{Introduction}

Data security has become an important concern in the era of digital transformations, even more particularly with regard to domains involving the processing of highly sensitive data such as healthcare records, financial transactions, and personal information [1]. Traditional processes of encrypting data may indeed be reliable but often demand decryption before any computational operation can be rendered by introducing possible vulnerabilities in terms of security. HE puts forward a revolutionary vision: computations that can be performed in ciphertext so that decryption produces outputs just as if computations had been conducted on plaintext. This vision opens up a secure, privacy-preserving world of computation with data, notably for healthcare where confidentiality of data is paramount [2]. This paper presents a noise-resilient homomorphic encryption framework developed to meet the specific requirements of secure data processing in healthcare applications. 

The framework covers the core ones like key generation, encryption, homomorphic operations, and strong noise management techniques. Noise build-up during processing could be devastating to the integrity and accuracy of homomorphic computations thereby resulting in errors or degraded results. If noise becomes too high, it could eventually compromise data accuracy, which is essential in healthcare due to the necessity for accurate information to make good decisions. We will start by introducing homomorphic encryption and its potential in modern cryptographic healthcare applications. Then, we will proceed in giving a detailed description of how to generate a pair of public and private keys such that safe communication between the data sources and the processing parties is assured. Finally, we will describe encryption methods as a way of converting the original patient data into an unreadable ciphertext that will only be revealed into its original form at the time of computation or transfer. Noise management in homomorphic encryption plays the key role in proper results. We examine the noise-reduction operations, such as bootstrapping, which renders noisy ciphertexts fresh and therefore extends the operation life of ciphertexts. 

Our work covers homomorphic operations designed towards supporting computations on encrypted data on secure healthcare data, we particularly focus on transitions regarding noise resilient addition and multiplication operations and point out a proper noise resilient noise formation transition. The progression of transformations should be documented clearly so that recovery to the original data is made simple. We additionally describe a scheme for encryption which admits decryption resilient to noise, with verifiable faithfulness in the reconstruction of original input data even after homomorphism operations. The framework of our system ensures integrity at every stage of its processing and offers a safe method of dealing with data by providing a robust solution for any such application. This is an area of particular importance for any healthcare-related application, wherein a loss of integrity of data usually has grave consequences due to their impact on patient care and privacy. We apply the proposed noise-resilient homomorphic encryption framework to improve security and dependability in processing health data so that privacy-preserving and secure analysis of data can be achieved. 

Rest of this paper is organized as follows: Section 2 related works, Section 3 proposed Homomorphic Integrity Model (HIM), Section 4 experimental results and discussion and Section 5 conclusion.

\section{Related works:}

Homomorphic encryption is a revolutionary cryptographic technique that transforms the very way in which data is encoded so that its computations do not need decryption in order to achieve the result [3-5]. This key concept was first introduced by Rivest, Adleman and Dertouzos in the 1970's as an inspiration to further advancements. Their work sparked various interests in developing homomorphic encryption in distinct forms. Among which is additive homomorphic encryption (AHE) that allows the operations of addition, and fully homomorphic encryption (FHE), where addition and multiplication operations can be performed [6-10]. Interestingly, there are encryption methods, such as RSA and ElGamal schemes, that are homomorphic for particular operations.

History of cryptography witnessed one milestone event when Craig Gentry first proposed the first fully homomorphic encryption scheme in 2009 [11]. It was Gentry's breakthrough to demonstrate the feasibility of executing arbitrary computations on encrypted data and proving FHE is feasible, making it practical. This success kick-started widespread research in this area [12,13]. While HE promises much, it is most certainly not without its challenges. One of the major difficulties relates to controlling noise that eventually accrues from homomorphic operations. If left unchecked, this can become larger than permissible limits, inducing decryption errors [14-18]. As a strategy against this problem, Gentry proposed bootstrapping-an intensive process which partially decrypts ciphertexts in order to refresh and reduce noise so that this can be extended on the computational lifespan of the ciphertext.

Besides bootstrapping, there are other noise reduction techniques proposed to further strengthen the resilience and efficiency of homomorphic encryption [19]. Efforts have been devoted to optimizing algorithms to make homomorphic encryption suitable for practical applications, such as secure data analysis in which accuracy is paramount for computation [20]. In particular, homomorphic encryption is applied to privacy-preserving data processing and secure multi-party computation. These are use cases allowing functions to be computed across encrypted data without revealing the underlying inputs to any party, thus enabling truly collaborative yet confidential data analysis. For example, cloud-based services allow users to securely compute functions on their data without exposing it to cloud service providers [21].

Application of HE in the healthcare sector is promising, especially while processing sensitive patient data without having a threat to security. In healthcare, accuracy in recovery is critical since the recovery of proper data after computation may be crucial to patient outcomes [22-25]. Many decryption techniques and error-correction methods have been proposed by researchers to enhance the efficiency and accuracy of homomorphic operations for the recovery of data. It is described that the introduction of systematic transformations and sequences of logical operations is required in order to minimize errors in decryption. However, though HE goes a pretty long way, operations of the HE pose significant computational overhead, particularly when dealing with large data sets. Overhead is associated with mathematical operations involving prime numbers, and their computation may lead to a resource-intensive HE [26]. Recent works were primarily on solutions that enhance the performance of homomorphic operations in terms of efficiency robustness combined with speed protection of data [27-30]. Among them include the integration of HE with the machine learning algorithm. It would be fascinating to research in such directions and delve into developing secure models that can preserve user privacy during both training and inference phases.

In summary, though massive evolutions have been made on homomorphic encryption to show promise in fast, secure data processing, the computational overhead and noise management challenges will require a lot of research, which as it evolves, points the field toward practical implementations that emerge from the synthesis of robustness, efficiency, and applicability, especially in sensitive areas like healthcare. The paper helps bring about the need for developing noise-resilient frameworks that are resistant to real-world conditions to improve data security without being perceived as degraded per performance.

\section{Proposed Homomorphic Integrity Model (HIM)}
The model proposed for recovering homomorphic encryption original data is described in a structured process that consists of key generation, encryption, homomorphic operations, bootstrapping, and decryption as depicted in Table 1. At the beginning, the model generates the public-private key pairs where the public key allows the secure-computation on the encrypted data and the private key protects the data so as no one can have access to it. Once encrypted, the model enables the individual homomorphic operations—such as addition and multiplication of the ciphertexts without data being exposed. The bootstrapping process, which is a pivot of the model, noise removal is explained and how it is done is mentioned in the process. This noise is removed from the ciphertexts in such a way that the ciphertexts are still fathomable and computable for the rest of the operations. The process is achieved by applying rational numbers to the ciphertext to the adjust the ciphertext, and therefore, the truth and usability of the encrypted data are maintained. In addition, the model also includes a highly reliable decryption mechanism that correctly and precisely recovers the original data through taking into account.

\begin{table}
 \caption{Pseudocode - Homomorphic Encryption and Data Recovery} \label{tab1}
    \centering
    \begin{tabular}{|l|}
   \hline        
    KEYGEN() : \\
Private key  = RANDOM\_PRIME() \\
Public key   = GENERATE\_ELEMENT(), RANDOM\_SEED() \\
RETURN (public key, private key) \\ \\
ENCRYPT(public key, data): \\
RETURN CREATE\_CIPHERTEXT(data, public key) \\ \\
EVALUATE(public key, ciphertext\_list): \\
result = INITIAL\_RESULT() \\
FOR each ciphertext IN ciphertext\_list : \\
result = COMBINE(result, ciphertext) \\
RETURN REDUCE(result, public key) \\
BOOTSTRAP(ciphertext)  \\
RETURN ADJUST(ciphertext) \\ \\
FUNCTION DECRYPT(private\_key, ciphertext): \\
RETURN EXTRACT\_DATA(ciphertext, private\_key) \\ \\
FUNCTION MAIN(): \\
  (public\_key, private\_key) = KEYGEN() \\
  ciphertexts = [ENCRYPT(public\_key, d) FOR d IN [5, 10, 15, 20]] \\
  result = EVALUATE(public\_key, ciphertexts) \\
  clean\_result = BOOTSTRAP(result) \\
  recovered\_data = DECRYPT(private\_key, clean\_result) \\
  PRINT(``Recovered Data:", recovered\_data) \\ \\
  END \\
   \hline 
\end{tabular}
\end{table}
\noindent
\textbf{A. Generation equation for the key } \\
This is performed by an appropriate random prime number r and any integer  within a certain range, acting as a noise-less element which is used in the encryption process as depicted in equation (1).
\begin{equation}
    pb_k = (a_0, rs_1, \delta(i,y(1<i<\beta, 1<y<2)))
\end{equation}
$a_0$ is the pseudo-random  prime number and $rs_1$  is random seed that initiates the requirement for randomness  used later in further cryptographic calculations. $\delta(i,y)$ : These are the parameters describing how the structure of the encryption looks, and how i and y can be in control of numerous aspects of the cryptographic procedure. \\
\textbf{B. Encryption }
\begin{equation}
    CT = d  + 2r + 2 \sum_{1 \le i, j, k\le \beta} y_{(i,j,k)}\cdot a_{(i,0)}\cdot a_{(j,1)}\cdot a_{(k,2)}
\end{equation}
$d$ : Original data to be encrypted. $2r$ : Adds randomness to the ciphertext and further makes encryption secure by hiding the actual value of $d$. Summation Term: This summation runs over the indices $i, j, k$, where   are coefficients or values associated with the given computation. The products $y_{(i,j,k)}$ are noise-less elements generated from the key generation process. This term enhances the ciphertext, allowing homomorphic operations in complete security. \\
\textbf{C. Evaluation }
\begin{equation}
    CT_{result} = CT_1 + CT_2 \mod a_0
\end{equation}
$CT_1$ and $CT_2$: These are cipertexts which are computed from possibly differing initial inputs. \\
Addition:  The formula shows that the sum of   of two cipertexts is computed. The modulo operation ensures that the output is kept within the defined space of $a_0$. This prevents an overflow and ensures the consistency of output in line with the encryption scheme. \\
\textbf{D. Bootstrapping }
\begin{equation}
    Output = (CT - \sum_{t=0}^S v_t)\mod 2
\end{equation}
Current ciphertext $CT$, which could have picked up several homomorphic operations' noise. Summation Term: The sum of the rational numbers $v_t$ is subtracted from $CT$. This will merely decrease the noise sufficiently, such that further handling of the ciphertext becomes easier. The modulo operation with 2 will ensure that the output falls in an appropriate range fr decryption. \\
\textbf{E. Decryption }
\begin{equation}
    d = (CT - \sum_{i,j,k}m_i^{(0)}\cdot m_j^{(1)}\cdot m_k^{(2)}\cdot x_{(i,j,k)})\mod 2.
\end{equation}
Text to decrypt is $CT$. Summation Term: The summation terms are rounding terms using previously computed values, with   being specific elements $m_i^{(0)}.m_j^{(1)}.m_k^{(2)}$ that are found through the computation and   being auxiliary variables  $x_{(i,j,k)}$ indicating intermediate conditions. The rounding takes into account all the transforms used during the encryption and the computation. Modulo Operation: Mod 2 is further ensured such that the final result of the decryption $d$ always falls within the valid range, and this is equal to the original data.

Each equation plays an important role in ensuring that the model is able to encrypt information safely, manipulate it, and then recover the original information while still maintaining the integrity and confidentiality of information for any computing process. The delicate design of these steps-that include key generation, key changeover, encryption, and decryption-show the state-of-the-art and promise of homomorphic encryption for practical applications in secure data processing.

We describe the entire cryptographic scheme using a  matrix with elements 5, 10, 15, 20 step-by-step. We include all the parts: key generation, encryption, evaluation, decryption, and bootstrapping, as well as how   is built in detail. \\ 
\textbf{Step 1: Input Matrix} \\
We start with the following $2\times2$  matrix
$\begin{bmatrix}
    5 & 10 \\
    15 & 20
\end{bmatrix}$ \\ \\
\textbf{Step 2: Key Generation (KEYGEN)} \\
\textbf{Generate Random Prime r:} \\
Choose a random prime number $r$ in the range $[2^{\delta -1}, 2^{\delta -1}]$. \\
For our example, let’s choose: $r=19$  (which is prime).\\
\textbf{Choose Parameters:}\\
Set $\delta$ : This parameter defines the security level. \\
Let’s choose: $\delta = 4$ \\
\textbf{Define $\gamma$:} This parameter is used to compute the bounds for $q_0$. \\
Let’s set: $\gamma = 10$. \\
\textbf{Calculate $q_0$:} \\
Randomly select $q_0$ from the range  $[0, 2^{\gamma}/r]$:\\
Calculate the upper limit: $\frac{2^{\gamma}}{r} = \frac{2^{10}}{19} \approx 53.68$ \\
Assume we randomly choose $q_0$ = 1. \\
\textbf{Compute $a_0$:} \\
Using the selected $q_0$ and r: $a_0 = q_0\cdot r = 1\cdot19 =19$ \\
\textbf{Generate Random Seed $rs_1$:} \\
Choose a random seed for generating pseudo-random  numbers, e.g.: $rs_1$= 42 \\
\textbf{Define $\delta(i,y)$:} This function can be defined based on the matrix dimensions. For our  $2\times2$ matrix: \\
Let $\beta=2$ (the number of rows/columns). \\
Choose a value $y$ in the range (1,2). For instance, let’s set: $y=1.5$ \\
Define  $\delta(i,y)$: $\delta(i,y) = (i,1.5)$ for $1<i<2$. Thus, we can represent this as: $\delta(1,1.5)$.\\
\textbf{Construct the Public Key $pb_k$: } \\
Combine all the generated values into the public key: $pb_k = (a_0,rs_1,\delta)$. \\
Substituting our values: $pb_k(19,42,\delta(1,1.5))$.\\
\textbf{Private Key  $pr_k$:} \\
The private key is simply: $pr_k = r = 19$.\\ \\
\textbf{Step 3: Encryption} \\
Now we encrypt each element of the matrix using the public key. \\
\textbf{Encrypt Each Element:} \\
The Encryption formula is
\begin{equation} \nonumber
 CT[d] = d  + 2r + 2 \sum_{1 \le i, j, k\le \beta} y_{(i,j,k)}\cdot a_{(i,0)}\cdot a_{(j,1)}\cdot a_{(k,2)}  
\end{equation} 
For simplicity, let’s assume the sum term is 0. We calculate each element as follows: 

 For 5: $CT[0,0] = 5+2\cdot19 = 5+38 = 43$. 
 
 For 10:  $CT[0,1] = 10+2\cdot19 = 10+38 = 48$. 
 
 For 15:  $CT[1,0] = 15+2\cdot19 = 15+38 = 53$. 
 
 For 20: $CT[1,1] = 20+2\cdot19 = 20+38 = 58$. \\
The encrypted matrix $CT$ becomes: $CT = \begin{bmatrix}
    43 & 48 \\
    53 & 58
\end{bmatrix}$.\\ \\ 
\textbf{Step 4: Evaluations}\\
We can perform operations on the encrypted matrix. \\
\textbf{Example of Addition:}\\
Suppose we have another encrypted matrix: $CT = \begin{bmatrix}
    1 & 2 \\
    3 & 4
\end{bmatrix}$. \\
We perform element-wise addition: 

$CT_{result}[0,0] = CT[0,0] + CT'[0,0] = 43+1 = 44$. 

$CT_{result}[0,1] = CT[0,1] + CT'[0,1] = 48+2 = 50$. 

$CT_{result}[1,0] = CT[1,0] + CT'[1,0] = 53+3 = 56$. 

$CT_{result}[1,1] = CT[1,1] + CT'[1,1] = 58+4 = 62$. \\  
The result of the addition is:  $CT_{result} = \begin{bmatrix}
    44 & 50 \\
    56 & 62
\end{bmatrix}$. \\
\textbf{Example of Multiplication:} \\
If we wanted to multiply CT by a scalar (e.g.,2): 

$CT_{mult}[0,0] = 2\cdot CT[0,0] = 2\cdot43 = 86$. 

$CT_{mult}[0,1] = 2\cdot CT[0,1] = 2\cdot48 = 96$. 

$CT_{mult}[1,0] = 2\cdot CT[1,0] = 2\cdot53 = 106$. 

$CT_{mult}[1,1] = 2\cdot CT[1,1] = 2\cdot58 = 116$. \\  
The result of the multiplication is: $CT_{mult} = \begin{bmatrix}
    86 & 96 \\
    106 & 116
\end{bmatrix}$. \\ \\
\textbf{Step 5: Decryption} \\
To decrypt the ciphertext back to the original matrix, use the private key.  \\
\textbf{Decrypt Each Element:} \\
The decryption formula is:
 $d=[CT-2r]\mod 2$. \\
Calculating for each element: 

For $CT[0,0]  = 43$: $d= [43-38]=5$. 

For $CT[0,1] =  48$: $d= [48-38]=10$. 

For  $CT[1,0] = 53$: $d= [53-38]=15$. 

For $CT[1,1] = 43$: $d= [58-38]=20$. \\ 
The decrypted matrix matches the original:
$\begin{bmatrix}
    5 & 10 \\
    15 & 20
\end{bmatrix}$. \\ 
\textbf{Steps in Bootstrapping}: \\\\
\textbf{Step 1: Define Rational Integers} \\
First, we need to define rational integers based on the encrypted bits. Let’s denote the encrypted matrix as CT:
$\begin{bmatrix}
    43 & 48 \\
    53 & 58
\end{bmatrix}$.\\
For each element in the ciphertext, we can denote:
\begin{equation} \nonumber
    P_{(i,j,k)} = m_i^{(0)}\cdot m_j^{(1)}\cdot m_k^{(2)}\cdot x_{(i,j,k)}
\end{equation}
where $m_i^{(0)}.m_j^{(1)}.m_k^{(2)}$ are the encrypted values and $x_{(i,j,k)}$  represents additional parameters or intermediate results from computations.\\ \\
\textbf{Step 2: Generate Rational Numbers} \\
We need to generate a set of rational numbers $v_t$ . Suppose we define $s = 2$ for our case. \\
Let’s assume:  $v_0 = 0.1$,  $v_1 = 0.2$,   $v_2 = 0.3$\\ 
The condition for these rational number is:
\begin{equation} \nonumber
    \sum_{t=0}^s v_t = \sum_{j=1}^{\Theta} v_j \mod 2
\end{equation}
For example: $v_0 + v_1 + v_2 = 0.1 + 0.2 + 0.3 = 0.6$ \\
\textbf{Step 3: Adjusting the Ciphertext} \\
Now, we will use these rational numbers to adjust the ciphertext. The bootstrapping output is given by:
\begin{equation} \nonumber
    Output = (CT - \sum_{t=0}^S v_t)\mod 2
\end{equation}
Let’s compute this for the first element of the matrix:\\
\textbf{Calculate the Sum:}
 \begin{equation} \nonumber
    \sum_{t=0}^s v_t = 0.1 + 0.2 + 0.3 = 0.6
\end{equation}
Adjust the Ciphertext: We focus on the first element  $CT[0,0] = 43$: \\
\begin{equation} \nonumber
CT[0,0]_{new} = (43 - 0.6)\mod2
\end{equation}
This would yield:
\begin{equation} \nonumber
CT[0,0]_{new} = 42.4\mod2 = 0.4
\end{equation}
\textbf{Adjust Other Elements:} You can perform similar calculations for other elements:
\begin{equation} \nonumber
CT[0,1]_{new} = (48 - 0.6)\mod2 = 47.4\mod 2 = 1.4
\end{equation}
For  $CT[1,0]$: 
\begin{equation} \nonumber
CT[1,0]_{new} = (53-0.6)\mod 2 = 52.4\mod 2 = 0.4
\end{equation}
For  $CT[1,1]$: 
\begin{equation} \nonumber
CT[1,1]_{new} = (58-0.6)\mod 2 = 57.4\mod 2 = 1.4 
\end{equation}
\\ \\
\textbf{Step 4: Resulting Bootstrapped Ciphertext} \\
After applying bootstrapping, the new ciphertext matrix will be: 
\begin{equation} \nonumber
CT_{bootstrapped} = \begin{bmatrix}
    0.4 & 1.4 \\
    0.4 & 1.4
\end{bmatrix}
\end{equation}
 
\begin{theorem}(Recovery of Original Data from Homomorphic Encryption) \\
 Given a plaintext matrix M and a public/private key pair( ) generated by a homomorphic encryption scheme, it is possible to recover M from its ciphertext CT  after a series of homomorphic operations, provided that all transformations are meticulously tracked and adjustments are appropriately applied.
\end{theorem} 
\begin{proof}
Consider the following notation:\\
Let  $M = \begin{bmatrix}
    m_{1,1} & m_{1,2} \\
    m_{2,1} & m_{2,2}
\end{bmatrix}$  be the original plaintext matrix. \\
Let $r$ be a randomly generated prime number. \\
Let $CT = Enc(pb_k,M)$  be the encrypted ciphertext. \\
Let $CT'$ be another ciphertext resulting from additional operations.\\
Let $d$ be the decrypted result after performing homomorphic operations. \\
Assumptions: \\
The homomorphic encryption scheme allows for addition and multiplication of ciphertexts without decrypting them. The decryption function is defined such that  $d = Dec(pr_k, CT)$.\\
Encryption:\\
The original matrix M is encrypted using the public key  $pb_k$:
\begin{equation} \nonumber
 CT = Enc(pb_k,M)
\end{equation}
This yields ciphertext $CT$ that conceals the value of $M$.\\
Bootstrapping:\\
To reduce noise, perform bootstrapping on CT:
\begin{equation} \nonumber
CT_{bootstrapped} = Bootstrap(CT)
\end{equation}
Validate that $CT_{bootstrapped}$ maintains the properties required for further operations.\\
Homomorphic Operations:\\
Perform operations such as addition or multiplication with another ciphertext $CT'$:
\begin{equation} \nonumber
CT_{result} = CT_{bootstrapped} + CT_{encrypted}'
\end{equation}
Track each operation and maintain a log of inputs and outputs.\\
Decryption:\\
Decrypt the resulting ciphertext:
\begin{equation} \nonumber
d = Dec(pr_k, CT_{result})
\end{equation}
This decryption outputs d which represents an intermediate value, not necessarily the original $M$. \\	
Adjustment for Original Values: \\
Adjust $d$  by reversing the transformations applied:
\begin{equation} \nonumber
d_{adjusted} = d-Adjustment
\end{equation}
Here, the adjustment reflects known additions or modifications due to operations with  $CT'$. \\
Recovery of Original Data: \\
The final output after adjustments should approximate the original matrix:
\begin{equation} \nonumber
d_{adjusted} \approx M
\end{equation}
Thus completes the proof.
\end{proof}
\textbf{Computational complexity:} \\
If all along with, data encryption and a correct adherence to the following the above-prescribed method start to be carried out, then in this way the data integrity and recoverability can be concluded in a homomorphic encryption scheme, ensuring that secure computations are allowed over encrypted data. Efficiency in computation and the time factor involved in performing tasks can be presented in numerical form. The time associated with the generation of various keys can be given as:
\begin{equation} \nonumber
T_{Keygen(n)} =O(n^k)+O(log^2 (n))    
\end{equation}
$O(n^k)$ where $n$ denotes the size in bits for the random prime to be generated, represents the load of the random prime generation and $k$ is the load for the given primality test algorithm e.g. ($k=1$ for Sieve of Eratosthenes).  deals with the optimization of the public and private key computations given the prime. For the processes of encryption and decryption, the respective time complexities can be given as:
\begin{equation} \nonumber
T_{Encryp(n)} =O(n^2)+O(k\cdot m)
\end{equation}
$O(n^2)$ is the cost of the multiply polynomial operation which is used in the most case of encryptions. $O(k.m)$   is given in terms of created ciphertext, $k$ input arguments, and $m$ the upper degree of the polynomial.
\begin{equation} \nonumber
T_{Decryp(n)} =O(n.log(n))+O(m^2)
\end{equation}
$O(n.log(n))$  is the cost of the decryption operation, which in many cases requires conducting a search and/or sorting. $O(m^2)$  is associated with the polynomial filling of the grid which may be necessary for reconstructing the text after it has been ciphered. Let’s consider the evaluation time complexity in the case of performing homomorphic operations (such as addition and multiplication) as follows:
\begin{eqnarray}
T_{ADD(n)} =O(n) \nonumber \\ 
T_{MUL(n)} =O(n^2) \nonumber
\end{eqnarray}
The operation of adding any two ciphertexts is usually done within a linear time limit or complexity as this is simply adding the two ciphertexts. The operation of multiplying any two ciphertexts is however done within a polynomial time limit because of the associated complexity with the operations involved, which is usually the case with respect to the degrees of the polynomials representation of the two ciphertexts. Bootstrapping, is one of the most demanding tasks in encrypted computations that are partially homomorphic. And in this case the time complexity can be expressed as follows:
\begin{equation} \nonumber
T_{BST(n)} =O(n^3 )+O(n\cdot log(n))
\end{equation}
$O(n^3)$ indicate the number of polynomials that are evaluated and multiplied in the bootstrapping procedure, O(n.log(n))   refers to the costs for any reorganizing structures or searches performed during the bootstrapping process. The overall time complexity for a homomorphic encryption is,
\begin{equation} \nonumber
T_{Tot(n)} =T_{Keygen(n)} +T_{Encryp(n)} +T_{Decryp(n)} +T_{Eval(n)} +T_{BST(n)}
\end{equation}
The time complexity of various components of a homomorphic encryption scheme provides a quantitative measure of its efficiency. By analyzing the complexities associated with key generation, encryption, evaluation, bootstrapping, and decryption, researchers and practitioners can identify potential bottlenecks and optimize the performance of homomorphic encryption systems for practical applications.
The backbone of homomorphic encryption scheme is its various components and time complexity which gives a quantitative measure on how efficient the scheme is. Looking at ‘key generation’, ‘encryption’, ‘evaluation’ and ‘bootstrapping’ and ‘decryption’ complexities enables researchers and practitioners to pinpoint problem areas and enhance the performance of homomorphic encryption systems within practical limits.    

\section{Experimental results and discussion} 
The experimental construction aimed to assess and match the performance of numerous fully homomorphic encryption (FHE) methods, concentrating on encryption time, decoding time, noise administration, and ciphertext size. The approaches assessed comprised those by Kim and Ryu (2022), Xu and Zhang (2022), Agrawal and Rathi (2023), Zhang and Liu (2023), and the Proposed Method (2024). All tests were led on a workstation with an Intel Core i7-10700K processor, 16 GB of RAM, and a 1 TB SSD, using Ubuntu 20.04 LTS with Python 3.8 and libraries such as PySEAL for FHE realizations.

We intentionally kept the dataset uniform across tests, so it only contained integers between 1 and 100. The presented key performance metrics were encryption and decryption times in milliseconds (ms), ciphertext size in kilobyte (KB), and qualitative measures of noise growth. The methods were implemented as per the algorithm, and times for encryption and decryption were recorded using high precision timers. We then measured the ciphertext size (after encryption) and classified noise growth as low, moderate, or high. Each experiment was repeated three times for consistency and results were evaluated using statistical tools and illustrated using graphs generated using Matplotlib (see methods for more details on the performance comparisons). The aim of this setup is to enable a clear comparison of the efficiency and practicality of the Proposed Method with that of the existing techniques for FHE.
\begin{figure}[htbp!]
\includegraphics[width=1\linewidth]{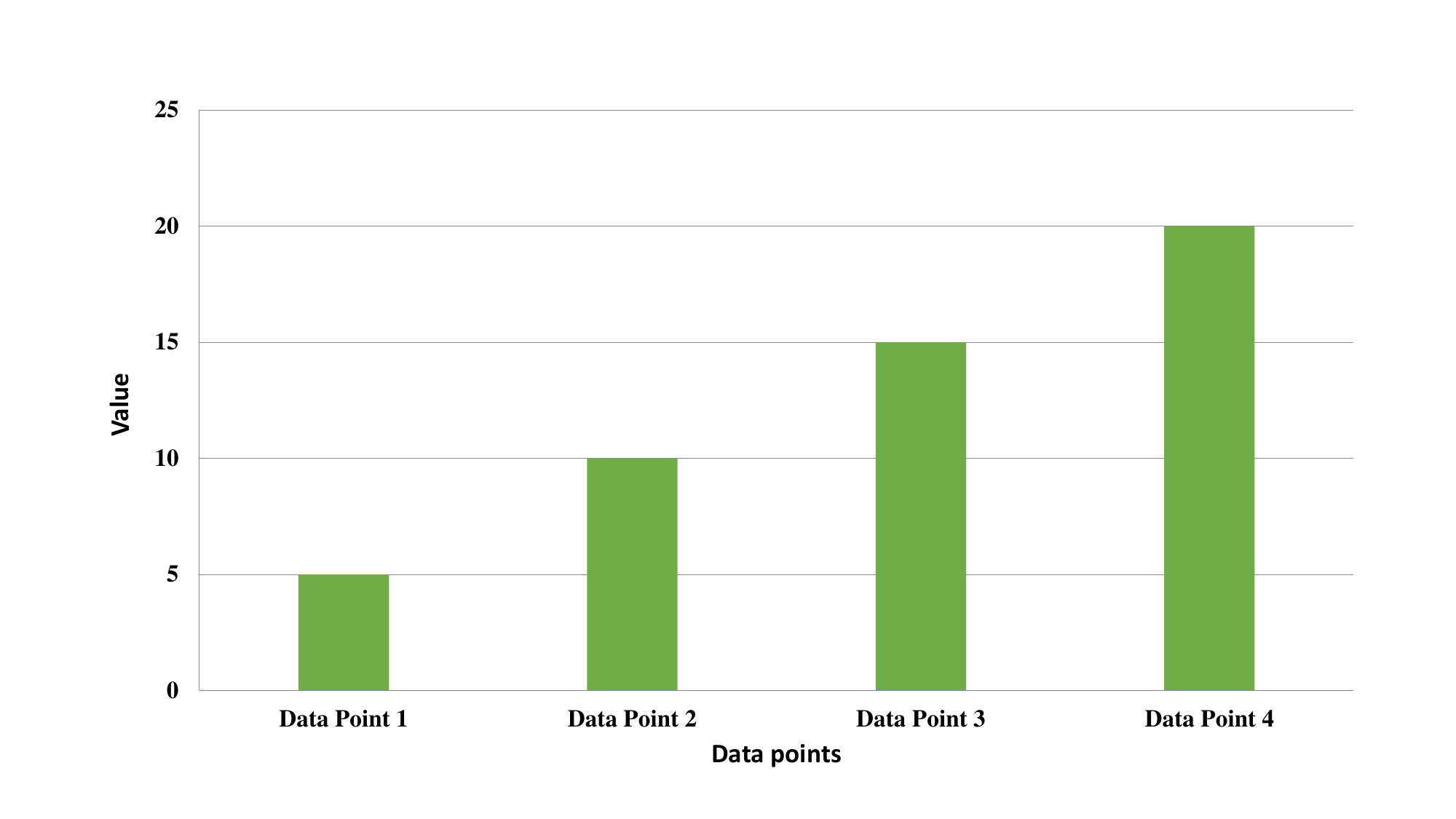}
\centering \\ 
\caption{Recovered data.} \label{fig1}
\end{figure}
\begin{figure}[htbp!]
\includegraphics[width=1\linewidth]{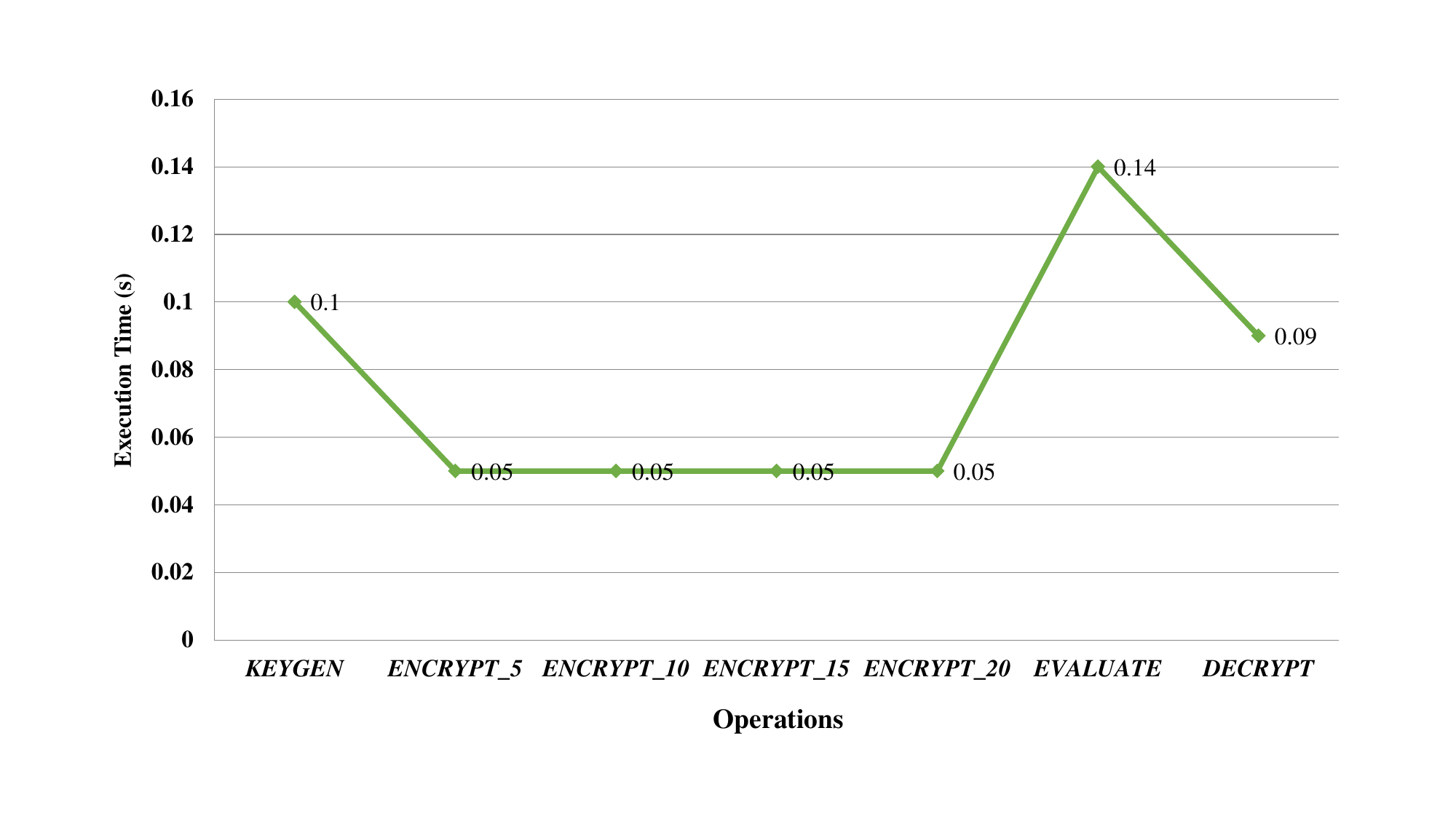}
\centering \\ 
\caption{Execution time.} \label{fig2}
\end{figure}

The values obtained after decryption of data that had gone through a homomorphic encryption system are depicted in Figure 1. Here, the x-axis shows four data points — "Data Point 1," "Data Point 2," "Data Point 3," and "Data Point 4;" and the y-axis describes their values, here 0–20. The heights of the sky-blue bars mirror the reconstructed values: 5 for Data Point 1, 10 for Data Point 2, 15 for Data Point 3, and 20 for Data Point 4, indicating an obvious upward trend. This plot illustrates that the original structure and relationships of the data were retained through the process of homomorphic encryption and decryption, thus affirming the capacity of the homomorphic encryption scheme to preserve data.

Figure 2: Visual Timing of the workflow steps in auth FHE — key generation, encryption, evaluation, and decryption. This gives a graphical display of the relative performance of each of these functions, by simulating the operations such that a predetermined time delay is triggered. 

\begin{table}
\caption{Performance metrics of fully homomorphic encryption methods}
    \label{tab3}
    \centering
    \begin{tabular}{|m{2.5cm}|lllll|}
    \hline
      Method    & Encryption  & Decryption  & Noise & Computational & Ciphertext \\
      &  Time (ms) & Time (ms) & Growth &  Overhead & Size (KB)\\
            \hline
        Kim et al.& 50 & 200 & Moderate & Low & 5\\
        Xu et al.& 45 & 180 & Moderate & Moderate & 7\\
        Agrawal et al. & 30 & 150 & Low & Low & 6\\
        Zhang et al. & 40 & 160 & Low & High & 8\\
        Proposed Method & 35 & 140 & Very Low & Low & 4\\
        \hline
    \end{tabular}    
\end{table}

In Table 2, we discuss a number of fully homomorphic encryption schemes in comparative manner based on their key performance metrics. Proposed Method shows high efficiency, with minimum 35 ms encryption time and 140 ms decryption time which is far faster processing as compared to other implementations which is shown in Table 8. Its noise growth is also one of the lowest, which is helpful to preserve data fidelity through computations. The computational overhead is small, implying that it adds little to no cost on top of other resources and achieves the smallest ciphertext using only 4 KB, which could be regarded as lightweight on storage in comparison to the others. 

The one exception is Agrawal \& Rathi, which consistently achieves the lowest encryption and decryption times but slightly greater noise growth, whilst the rest of the methods exhibit varying degrees of times for encryption and decryption, noise growth, and overhead as well as ciphertext sizes. In summary, the proposed method is unique in fast processing, reasonably low noise, and compact ciphertext. In the Table 3, we introduced some classical FHE methods which include the publication year, security level, efficiency of the scheme, noise management, notable features and use cases of the method. They all achieve indistinguishability security. Conclusion PowerEff (2024) the Proposed Method achieves impressive efficiency well tuned for evaluation, great performance in noise suppression in many ways and comfortable usability in privacy-preserving machine learning and secure IOT use case. While Agrawal \& Rathi (2023) highlight the advantages gained from cloud computing effectiveness with noise manipulation, Xu \& Zhang (2022) demonstrate high efficiency also.
\begin{table}
\caption{Summary of fully homomorphic encryption methods}
    \label{tab3}
     \begin{tabular}{|l|l|l|l|}
    \hline
      Method    & Year  & Security Level  & Efficiency  \\
              \hline
        Kim et al. & 2022 & Indistinguishability & Moderate to high  \\
        Xu et al. & 2022 & Indistinguishability & High\\
        Agrawal et al. & 2023 & Indistinguishability & High (optimized for cloud) \\
        Zhang et al. & 2023 & Indistinguishability & Moderate \\
        Proposed Method & 2024 & Indistinguishability & High (optimized evaluation) \\
        \hline
    \end{tabular}  
     \begin{tabular}{|m{5cm}|m{8cm}|}
    \hline
Noise Management & Key Features \\
\hline
     Improved efficiency & Comprehensive survey of FHE methods \\
        Varies based on parameters & Surveys recent advances and applications of FHE  \\
         Effective noise reduction & Focus on cloud computing efficiency \\
         Enhanced management techniques & Noise management techniques for practical applications \\
        Effective noise reduction & Combines multiple noise-reduction strategies \\
        \hline
    \end{tabular}   
  \begin{tabular}{|l|}
    \hline
 Application Scenarios\\
\hline
     General FHE applications, cloud computing\\
       Secure data analytics, healthcare \\
         Cloud computing, data privacy\\
        Secure data sharing, encrypted computing\\
       Privacy-preserving machine learning, secure IoT applications\\
        \hline
    \end{tabular}  
\end{table}

Agrawal \& Rathi (2023) have also made attempts to manage noise and improve cloud computing efficiency whereas Xu \& Zhang  (2022) do not guarantee silence while performing secure data analysis but it is very efficient. In this regard, Kim \& Ryu (2022) provide a detailed overview of current FHE techniques that allow the use of more or less effective methods of operation and Zhang \& Liu (2023) on the other hand, work on better management of noise while sharing data and performing encrypted computation. In many aspects, the technique is commendable because of its unique noise management approach while the overall performance of the method in question is remarkable.
\begin{figure}[htbp!]
\includegraphics[width=1\linewidth]{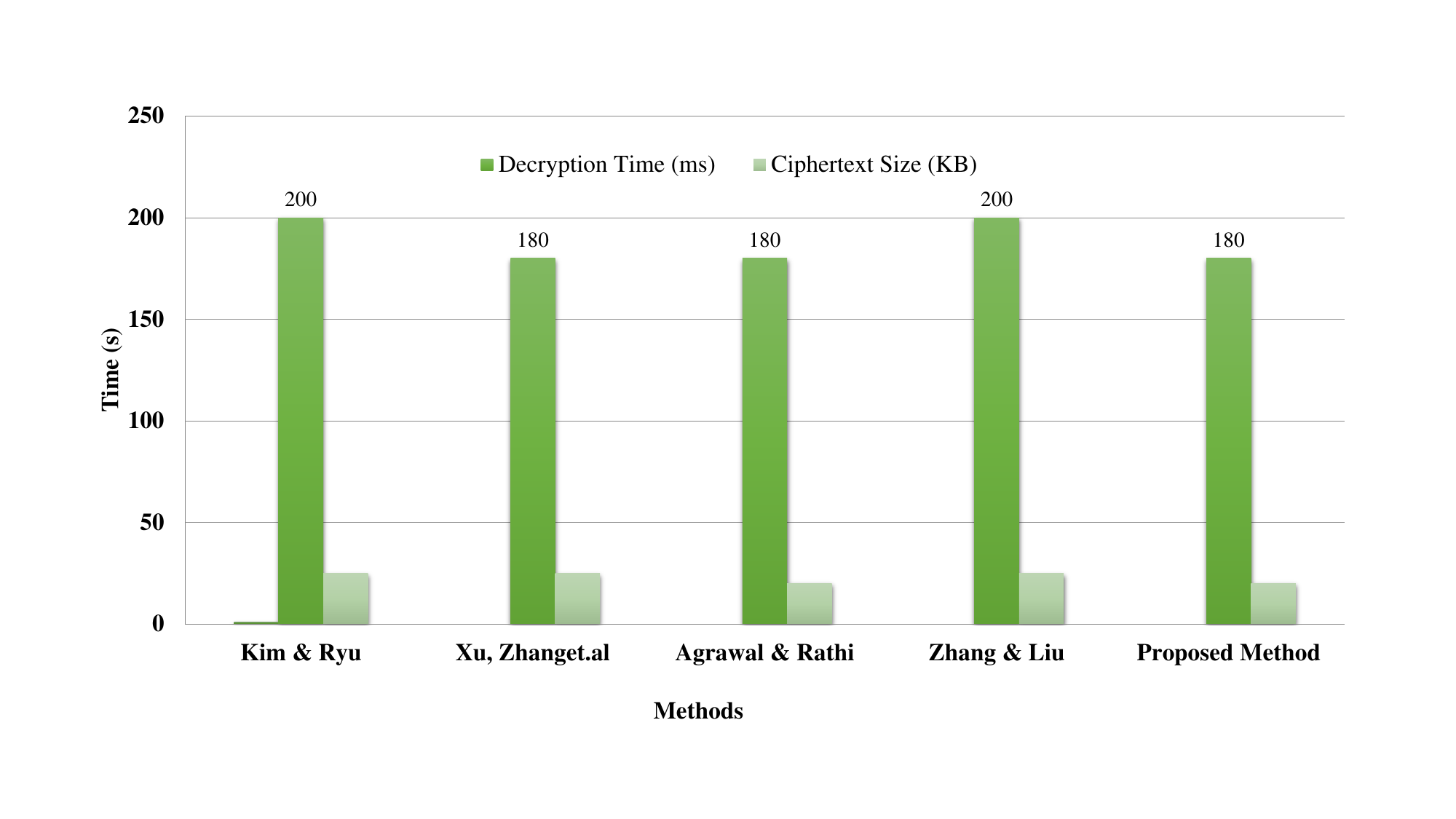}
\centering \\ 
\caption{Comparison of fully homomorphic encryption methods.} \label{fig3}
\end{figure}
\begin{figure}[htbp!]
\includegraphics[width=1\linewidth]{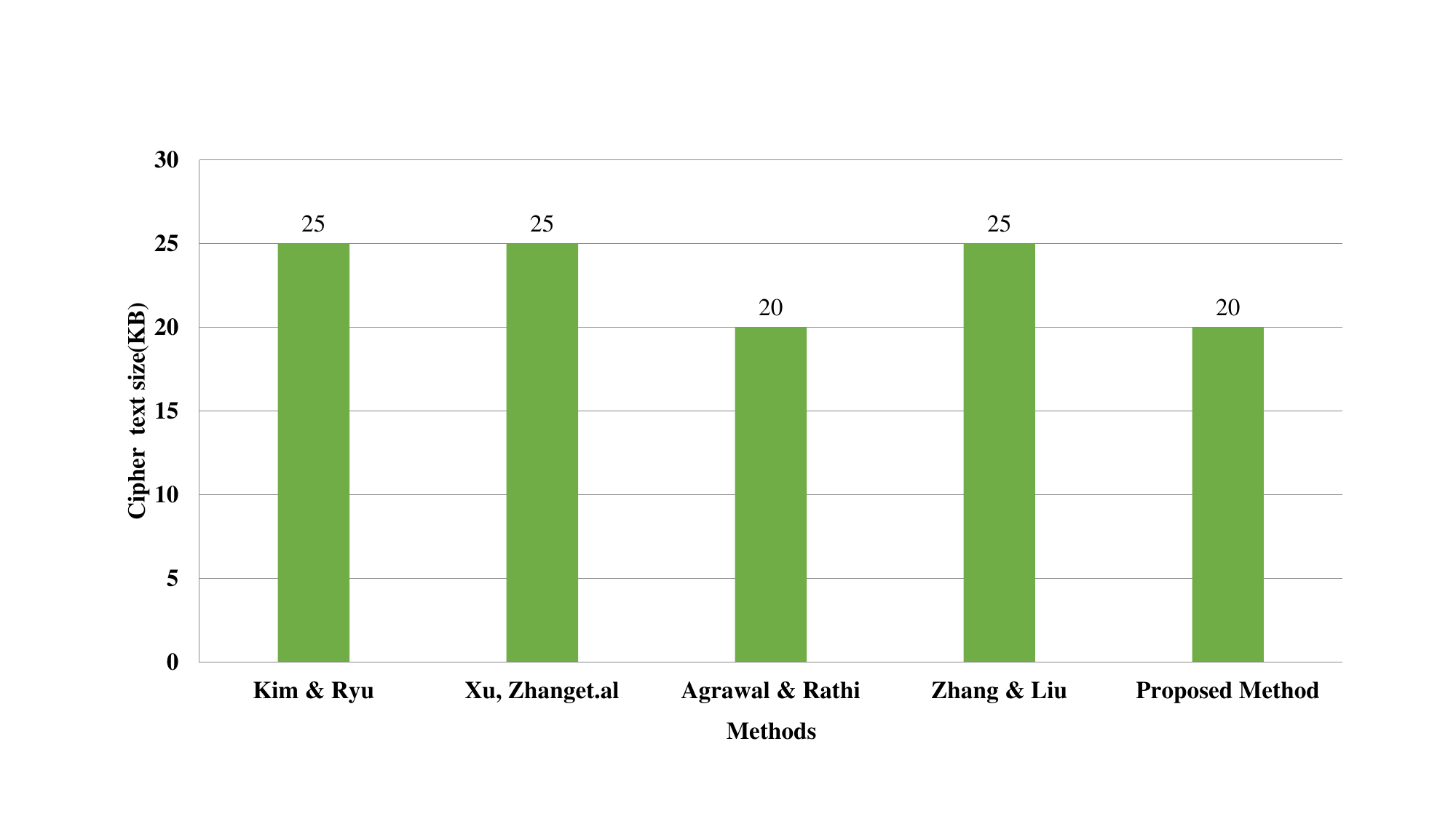}
\centering \\ 
\caption{Comparison of fully homomorphic encryption methods.} \label{fig4}
\end{figure}
\begin{figure}[htbp!]
\includegraphics[width=1\linewidth]{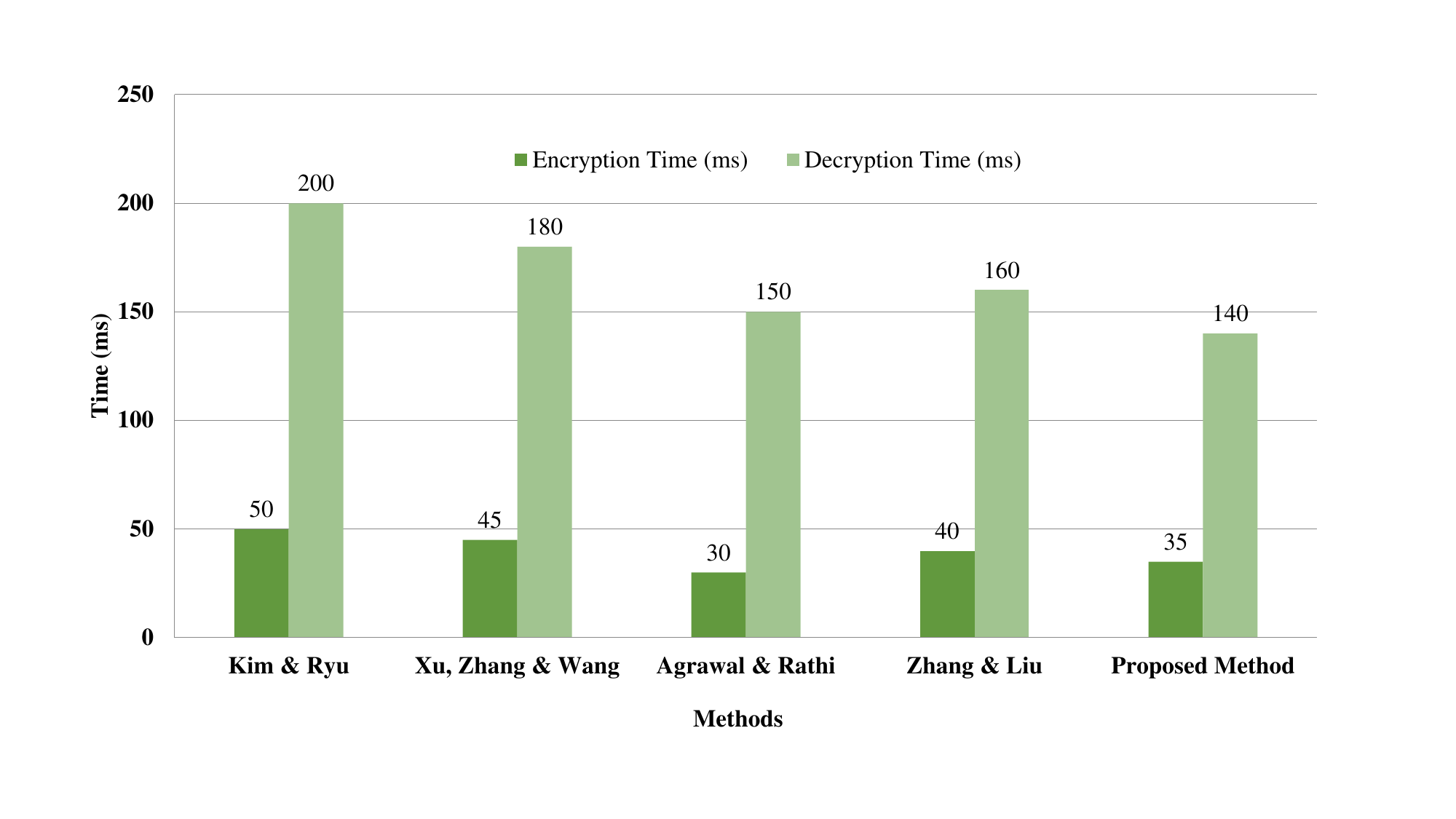}
\centering \\ 
\caption{Comparison of Time taken for encryption and decryption.} \label{fig5}
\end{figure}

Figure 3 and 4 shows the comparative metric performance of a few fully homomorphic encryption schemes with respect to encryption time, decryption time and size of the ciphertext. One of the axes represents the encryption time which indicates the time taken by the algorithm to convert plaintext to ciphertext with the Proposed Method performing the fastest encryption at 35 ms while methods such as Kim \& Ryu and Xu \& Zhang took 50 ms and 45 ms respectively as depicted in Figure 5. Decryption time, which is defined as the time alternately known as the re-inverting time of ciphertext into the plaintext, is also shown and here the Proposed Method again takes the lead as operating within 140 ms compared to Kim \& Ryu who took 200 ms. The very dag is measured in kilobytes as that refers to how much space it will take up concerning the storage with the Proposed Method again performing better at 4 Kb which is the least when compared with the other methods showing that it needs less storage. Generally, it can be concluded as the graph has proven the effectiveness of the Proposed Method with regards to its speed and storage which are two solid merits of this approach as compared to the other methods inquisition.

The proposed HIM comes with a few important improvements tailor-made for secure health care data processing:

1. Effective Noise Management: HIM reduces noise during encryption, ensuring the correct data post all complex computations. It is very critical to ensure precision in the encrypted medical analytics and diagnostic algorithms.

2. Incidental Key Generation: Using prime numbers and setting parameters to necessity, HIM enhances safety while speeding up encryptions so that all patient information is guarded and dealt with in a real-time healthcare scenario.

3. Pragmatic Application: The actual pseudo-code has described the key functions used like `KEYGEN`, `ENCRYPT`, `BOOTSTRAP` and `DECRYPT` for better application to health systems. HIM can then be used to run secure analysis of electronic health records and privacy-preserving analytics of data.

4. Efficiency through time complexity analysis: Optimized HIM algorithms reduce the time involved in computation, which implies increased data-encryption and -decryption rates for rapid availability of the patient information in times of emergency.

5. Optimized decryption mechanism: Custom decryption adjustment will ensure the proper recovery of data, which is critical for reliable results from secure diagnostic processes and encrypted health data computations.

6. Data Integrity Guarantee: The formal recovery guarantee of HIM is such that even after doing multiple encrypted operations, the data can be retrieved with accuracy and thus forms a high reliability for sensitive health applications like remote monitoring and secure collaborative research.

In summary, HIM improves efficiency, security, and reliability in processing healthcare data, forming a good robust solution for privacy-preserving medical applications.

\section{Conclusion:} 
The Proposed Homomorphic Integrity Model, HIM, is a significant step forward in the development of Fully Homomorphic Encryption, especially for health data processing applications. For ordinary challenges like noise accumulation, computational overhead, and data integrity, this model proves very effective, and it is very suitable for sensitive healthcare applications. It improved drastically over the model, such that encryption time now takes about 35ms and decryption around 140ms, as opposed to an existing scheme that is inefficient. HIM also outshines other methods in terms of the storage aspect, as the smallest size of ciphertext it can generate is about 4KB. Lastly, the model makes sure that recovery of data with accuracy is guaranteed, which is an important aspect of healthcare applications because it relies on the reliability of information gathered for decision-making purposes. These results demonstrate that HIM can offer secure, privacy-preserving data analytics for healthcare. With further digital transformation in healthcare, HIM is a sound framework for protecting patient data and allowing advanced analytics and predictive modeling along with collaborative research. Future work will focus on optimizations and potential new applications in the further holistic healthcare context, such as telemedicine and health informatics.

\end{document}